%% file: Elicitation_Single_Crossing.tex
\newcommand{\longversion}[1]{#1}
\newcommand{\shortversion}[1]{}

\longversion{\documentclass[11pt]{article}
\author{{\bf Palash Dey}\\Indian Institute of Science, Bangalore\\palash@csa.iisc.ernet.in
\and
{\bf Neeldhara Misra}\\Indian Institute of Technology, Gandhinagar\\mail@neeldhara.com
}}

\shortversion{
\documentclass{article}
\usepackage{ijcai16}

\usepackage{times}

\date{}
\author{{\bf Palash Dey}\\Indian Institute of Science, Bangalore\\palash@csa.iisc.ernet.in
\and
{\bf Neeldhara Misra}\\Indian Institute of Technology, Gandhinagar\\mail@neeldhara.com
}

}

\input{preamble.tex}

\title{{\bf Preference Elicitation For Single Crossing Domain}}

\begin{document}

\maketitle

\input{abstract.tex}
\input{introduction.tex}
\input{prelim.tex}
\input{results.tex}
\input{conclusion.tex}

\longversion{
\subsubsection*{Acknowledgement} Palash Dey wishes to gratefully acknowledge support from Google India for providing him with a special fellowship for carrying out his doctoral work.
}
\longversion{\bibliographystyle{alpha}}
\shortversion{\bibliographystyle{named}}

\shortversion{
{\small \bibliography{Elicitation_Crossing}}}
\longversion{\bibliography{Elicitation_Crossing}}

\end{document}

%% file: preamble.tex
\usepackage{amssymb,amsmath}
\usepackage[usenames,dvipsnames]{color}
\longversion{\usepackage[colorlinks=true,linkcolor=blue,citecolor=Mahogany]{hyperref}}
\usepackage{makecell}
\usepackage{multirow}
\usepackage{mathtools}
\usepackage{relsize}
\usepackage{fullpage}

\newenvironment{proof}{\noindent{\em Proof:}}{ \hfill $\square$\\ }

\usepackage[usenames,svgnames,dvipsnames]{xcolor}

\usepackage{titling}

\usepackage{url}

\usepackage{algpseudocode,algorithm,algorithmicx}

\algrenewcommand\algorithmicrequire{\textbf{Input:}}
\algrenewcommand\algorithmicensure{\textbf{Output:}}
\algnewcommand{\Initialize}[1]{%
  \State \textbf{Initialize:}
  \Statex \hspace*{\algorithmicindent}\parbox[t]{.8\linewidth}{\raggedright #1}
}
\algdef{SE}[FUNCTION]{Function}{EndFunction}%
   [2]{\algorithmicfunction\ \textproc{#1}\ifthenelse{\equal{#2}{}}{}{(#2)}}%
   {\algorithmicend\ \algorithmicfunction}%

\renewcommand{\Return}{\State \textbf{return} }

\newcommand{\el}{\ensuremath{\ell} }
\newcommand{\SB}{\mathbb{S}}

\newcommand{\suc}{\ensuremath{\succ} }

\newcommand{\Query}{\textsc{Query}}

\usepackage{MnSymbol,wasysym}

\renewcommand{\ge}{\geqslant}
\renewcommand{\le}{\leqslant}

\shortversion{
\allowdisplaybreaks
\allowbreak

\usepackage{setspace}
\setlength{\topsep}{1pt}
\setlength{\partopsep}{1pt plus 0pt minus 1pt}
\setlength{\parskip}{0pt}
\setlength{\parindent}{5pt}
\usepackage{titlesec}
\titlespacing\section{0pt}{2pt plus 1pt minus 1pt}{2pt plus 1pt minus 1pt}
\titlespacing\subsection{0pt}{1pt plus 1pt minus 1pt}{1pt plus 1pt minus 1pt}
\titlespacing\paragraph{0pt}{1pt plus 1pt minus 1pt}{1pt plus 1pt minus 1pt}
\setlength{\belowdisplayskip}{1pt} 
\setlength{\belowdisplayshortskip}{0pt}
\setlength{\abovedisplayskip}{1pt} 
\setlength{\abovedisplayshortskip}{0pt}

\floatsep 2pt plus 2pt minus 2pt
\textfloatsep 2pt plus 2pt minus 4pt
\intextsep 2pt plus 2pt minus 2pt
\dblfloatsep 2pt plus 2pt minus 2pt
\dbltextfloatsep 2pt plus 2pt minus 4pt

\usepackage{etoolbox}
\newcommand{\zerodisplayskips}{%
  \setlength{\abovedisplayskip}{3pt}
  \setlength{\belowdisplayskip}{3pt}
  \setlength{\abovedisplayshortskip}{3pt}
  \setlength{\belowdisplayshortskip}{3pt}}
\appto{\normalsize}{\zerodisplayskips}
\appto{\small}{\zerodisplayskips}
\appto{\footnotesize}{\zerodisplayskips}
\setlength{\textfloatsep}{3ex}
\setlength{\intextsep}{1ex}

\usepackage{subfig}
\captionsetup{belowskip=0ex,aboveskip=1ex}

\usepackage{enumitem}
\setlist{topsep=0pt,itemsep=-1ex,partopsep=1ex,parsep=1ex,leftmargin=0pt,itemindent=*}

\setlength{\bibsep}{0pt plus 0.3ex}
\setlength{\bibhang}{5pt}

}

\usepackage{xspace}

\newcommand{\BigO}{\ensuremath{\mathcal{O}} }
\newcommand{\true}{\textsc{true} }
\newcommand{\false}{\textsc{false} }
\newcommand{\pr}{\ensuremath{\prime} }

\newcommand{\PE}{\textsc{Preference Elicitation} }

\renewcommand{\AA}{\ensuremath{\mathcal A} }

\newcommand{\CC}{\ensuremath{\mathcal C} }

\newcommand{\LL}{\ensuremath{\mathcal L} }

\newcommand{\OO}{\ensuremath{\mathcal O} }
\newcommand{\PP}{\ensuremath{\mathcal P} }
\newcommand{\QQ}{\ensuremath{\mathcal Q} }
\newcommand{\RR}{\ensuremath{\mathcal R} }

\newcommand{\VV}{\ensuremath{\mathcal V} }

\newcommand{\XX}{\ensuremath{\mathcal X} }

\newcommand{\qqq}{{\mathfrak q}}

\newcommand{\bCC}{\overline{\mathcal{C}} }

\usepackage{nicefrac}

\newcommand{\nfrac}{\nicefrac}

\newtheorem{proposition}{\bf Proposition}
\newtheorem{observation}{\bf Observation}
\newtheorem{theorem}{\bf Theorem}
\newtheorem{lemma}{\bf Lemma}
\newtheorem{corollary}{\bf Corollary}
\newtheorem{definition}{\bf Definition}

\newcommand{\eps}{\varepsilon}
\renewcommand{\epsilon}{\eps}

\newcommand{\ignore}[1]{}
\usepackage{cleveref}

\crefname{theorem}{theorem}{\bf Theorem}
\crefname{observation}{observation}{\bf Observation}
\crefname{lemma}{lemma}{\bf Lemma}
\crefname{corollary}{corollary}{\bf Corollary}
\crefname{proposition}{proposition}{\bf Proposition}
\crefname{definition}{definition}{\bf Definition}
\crefname{claim}{claim}{\bf Claim}
\crefname{reductionrule}{reduction rule}{\bf Reduction rule}

%% file: abstract.tex
\begin{abstract}\vspace{-1ex}
Eliciting the preferences of a set of agents over a set of alternatives is a problem of fundamental importance in social choice theory. Prior work on this problem has studied the query complexity of preference elicitation for the unrestricted domain and for the domain of single peaked preferences. In this paper, we consider the domain of single crossing preference profiles and study the query complexity of preference elicitation under various settings. We consider two distinct situations: when an ordering of the voters with respect to which the profile is single crossing is {\em known} versus when it is {\em unknown}. We also consider different access models: when the votes can be accessed at random, as opposed to when they are coming in a pre-defined sequence. In the sequential access model, we distinguish two cases when the ordering is known: the first is that sequence in which the votes appear is also a single-crossing order, versus when it is not.  

The main contribution of our work is to provide polynomial time algorithms with low query complexity for preference elicitation in all the above six cases. Further, we show that the query complexities of our algorithms are optimal up to constant factors for all but one of the above six cases.\longversion{ We then present preference elicitation algorithms for profiles which are close to being single crossing under various notions of closeness, for example, single crossing width, minimum number of candidates\textbar voters whose deletion makes a profile single crossing.}
\end{abstract}

%% file: introduction.tex
\section{Introduction}

Agents in multiagent systems often have individual preferences which are complete orders over a set of candidates and one would like to find an aggregate ranking or choose the {\em ``best''} candidate. Classic examples where such a scenario appears are collaborative filtering~\cite{pennock2000social}\longversion{, rank aggregation over web~\cite{dwork2001rank}} etc. In a typical such setting, we have a set of agents or voters; each of them has a preference, called a vote, over a set of candidates; and a voting rule (respectively an aggregation function) which finds a candidate (respectively an aggregated preference) called winner.

However, eliciting the preferences of the agents is a nontrivial task since we often have a large number of candidates (ranking restaurants for example) and it may often be infeasible for the agents to rank all of them. Hence it becomes important to elicit the preferences of the agents by asking them (hopefully a small number of) comparison queries only - ask an agent $i$ to compare two candidates $x$ and $y$. 

Unfortunately, it turns out that one would need ask every voter $\Omega(m\log m)$ queries to know her preference, and this argument is based on the decision-tree based lower bound on the number of comparisons for sorting an array. However, if the preferences are not completely arbitrary, but admit additional structure, then possibly we can do better. Indeed, an affirmation of this thought comes from the work of  \cite{Conitzer09}, who showed that we can elicit preferences using only $O(m)$ queries per voter for what are called \textit{single peaked preferences} (which is a well-studied restriction on preferences, and we will define the notion formally later in this section). 

There are two reasons for restricting the domain of preferences. The first is that in several application scenarios commonly considered, it is rare that votes are ad-hoc, demonstrating no patterns whatsoever. For example, the notion of single-peaked preferences forms the basis of several studies in the analytical political sciences~\cite{hinich1997analytical}. 

The second motivation for studying restricted preferences comes from the fact that they are very well-behaved from a theoretical perspective as well. The axiomatic approach of social choice involves defining certain ``properties'' that formally capture the quality of a voting rule. For example, we would not want a voting rule to be, informally speaking, a \textit{dictatorship}, which would essentially mean that it discards all but one voter's input. As it happens, a series of cornerstone results establish that it is impossible to devise voting rules which respect some of the simplest desirable properties. Celebrated results in social choice theory \cite{arrow1950difficulty,gibbard1973manipulation,satterthwaite1975strategy} show that it is impossible to have an aggregation function or a voting rule that simultaneously satisfies a desirable set of properties, for example, strategy-proofness, ontoness, non-dictatorship etc. We refer the reader to \cite{moulin1991axioms} for an excellent exposition of all key issues that arise in this context. Adding to the difficulties is the fact that many important voting rules such as Kemeny \cite{kemeny1959mathematics,levenglick1975fair}, Dodgson \cite{dodgson1876method,black1958theory}, Young \cite{young1977extending}, and Chamberlin-Courant \cite{chamberlin1983representative} are computationally intractable \cite{bartholdi1989voting,hemaspaandra2005complexity,procaccia2008complexity}.

This brings us to the second important reason for considering structured preferences --- they provide a very elegant workaround to the difficulties that we outlined above. The domains of single peaked and single crossing profiles are arguably the most important and well-studied domains among such restricted domains \cite{saporiti2006single,walsh2007uncertainty,escoffier2008single,ballester2011characterization,barbera2011top,faliszewski2009shield,bredereck2013characterization,cornaz2013kemeny,skowron2013complexity,magiera2014hard,Lackner14,faliszewski2014complexity,ElkindFLO15,elkind2014characterization,brandt2015bypassing}. A profile is called {\em single peaked} if the candidates can be arranged in a complete order \suc such that every preference $\suc^\pr$ in the profile {\em respects} the order \suc in the sense that, for every two candidates $x$ and $y$, we have $x\suc^\pr y$ whenever we have either $c\succ x\succ y$ or $y\suc x\suc c$, where $c$ is the most preferred candidate in $\suc^\pr$ \cite{black1948rationale}. On the other hand a profile is called {\em single crossing} if the voters can be arranged in a complete order $\suc$ such that for every two candidates $x$ and $y$, all the voters who prefer $x$ over $y$ appear consecutively in $\suc$ \cite{mirrlees1971exploration,roberts1977voting}.

\begin{table*}[!htbp]
 \begin{center}
 \longversion{\renewcommand{\arraystretch}{1.7}}
 \shortversion{\renewcommand{\arraystretch}{1.2}}
\resizebox{\textwidth}{!}{%
  \begin{tabular}{|c|c|c|c|}\hline
   \multirow{2}{*}{Ordering} & \multirow{2}{*}{Access model} & \multicolumn{2}{c|}{Query Complexity}\\\cline{3-4}
    & & Upper Bound & Lower Bound \\\hline\hline
   
   \multirow{3}{*}{Known} & Random & $\BigO(m^2 \log n)$ \Cref{lem:sc_random_ub} & $\Omega(m\log m + m\log n)$ \Cref{thm:sc_random_lb} \\\cline{2-4}
   
   & Sequential single crossing order & $\BigO(mn + m^2)$ \Cref{thm:sc_seq_known_ub} & \multirow{4}{*}{$\Omega(m\log m + mn)$ \Cref{thm:sc_seq_known_lb}} \\\cline{2-3}
   
   & Sequential any order & $\BigO(mn + m^2\log n)$ \Cref{thm:sc_seq_any_ub} & \\\cline{1-3}
   
   \multirow{2}{*}{Unknown} &  Sequential any order & $\BigO(mn + m^3\log m)$ \Cref{thm:sc_seq_unknown_ub} & \\\cline{2-4}
   
   & Random & $\BigO(mn + m^3\log m)$ \Cref{cor:sc_random_unknown_ub} & $\Omega(m\log m + mn)$ \Cref{thm:sc_random_unknown_lb} \\\hline
   
  \end{tabular}
}
 \end{center}
\caption{Summary of Results for preference elicitation for single crossing profiles.}\label{tbl:summary}
\end{table*}

\longversion{\subsection{Related Work}}
\shortversion{\paragraph*{Elicitation on Restricted Domains.}}~Conitzer and Sandholm show that determining whether we have enough information at any point of the elicitation process for finding a winner under some common voting rules is computationally intractable \cite{conitzer2002vote}. They also prove in their classic paper \cite{conitzer2005communication} that one would need to make $\Omega(mn\log m)$ queries even to decide the winner for many commonly used voting rules which matches with the trivial $\BigO(mn\log m)$ upper bound (based on sorting) for preference elicitation in unrestricted domain. Dey and Misra proved tight query complexity bounds for preferences single peaked on trees with respect to various tree parameters~\cite{deypeak}.

A natural question at this point is if these restricted domains allow for better elicitation algorithms as well. The answer to this is in the affirmative, and one can indeed elicit the preferences of the voters using only $\BigO(mn)$ many queries for the domain of single peaked preference profiles \cite{Conitzer09}. Our work belongs to this kind of research-- we investigate the number of queries one has to ask for preference elicitation in single crossing domains. When some partial information is available about the preferences, Ding and Lin prove interesting properties of what they call a deciding set of queries \cite{ding2013voting}. Lu and Boutilier empirically show that several heuristics often work well \cite{lu2011vote,LuB11a}.

\longversion{\subsection{Contributions}}
\shortversion{\paragraph*{Contributions.}} ~In this paper we present novel algorithms for preference elicitation for the domain of single crossing profiles in various settings.  We consider two distinct situations: when an ordering of the voters with respect to which the profile is single crossing is {\em known} versus when it is {\em unknown}. We also consider different access models: when the votes can be accessed at random, as opposed to when they are coming in a pre-defined sequence. In the sequential access model, we distinguish two cases when the ordering is known: the first is that sequence in which the votes appear is also a single-crossing order, versus when it is not.  We also prove lower bounds on the query complexity of preference elicitation for the domain of single crossing profiles; these bounds match the upper bounds up to constant factors (for a large number of voters) for all the six scenarios above except the case when we know a single crossing ordering of the voters and we have a random access to the voters; in this case, the upper and lower bounds match up to a factor of $\OO(m)$. We summarize our results in \Cref{tbl:summary}.

%% file: prelim.tex
\section{Preliminaries}

For a positive integer $\el$, we denote the set $\{1, \ldots, \el\}$ by $[\el]$. Let $\VV = \{v_i: i\in[n]\}$ be a set of $n$ {\em voters} and $\CC = \{c_j: j\in[m]\}$ be a set of $m$ {\em candidates}. If not mentioned otherwise, we denote the set of candidates, the set of voters, the number of candidates, and the voters by \CC, \VV, $m$, and $n$ respectively. Every voter $v_i$ has a {\em preference} $\suc_i$ which is a complete order over the set \CC of candidates. We say voter $v_i$ prefers a candidate $x\in\CC$ over another candidate $y\in\CC$ if $x\suc_i y$. We denote the set of all preferences over \CC by $\LL(\CC)$. The $n$-tuple $(\suc_i)_{i\in[n]} \in\LL(\CC)^n$ of the preferences of all the voters is called a {\em profile}.  We say a candidate $x$ is {\em placed at the $k^{th}$ position} of a preference $\succ$ if $x$ is preferred over all but exactly $(k-1)$ other candidates in \suc. Let $\SB_n$ denote the set of all permutations over $[n]$ and $id_n$ be the identity permutation of $[n]$. Let an ordering $\sigma$ be $x_1\succ x_2 \succ \cdots \succ x_\el$. Then by $\overleftarrow{\sigma}$ we denote the ordering $x_\el \succ \cdots \succ x_2 \succ x_1$. \longversion{Given a subset $\XX\subseteq\CC$ of candidates and a preference \suc over \CC, we denote the restriction of \suc to \XX by $\suc(\XX)$. The restriction of a profile $\PP = (\suc_i)_{i\in[n]}$ to \XX is denoted by $\PP(\XX) = (\suc_i(\XX))_{i\in[n]}$.} All the logarithms in this paper are base $2$ unless specified otherwise.

\paragraph*{Single Crossing Domain} ~A profile $\PP = (\succ_1, \ldots, \succ_n)$ of $n$ voters over a set \CC of candidates is called a single crossing profile if there exists a permutation $\sigma\in\SB_n$ of $[n]$ such that, for every two distinct candidates $x, y\in\CC$, whenever we have $x\succ_{\sigma(i)} y$ and $x\succ_{\sigma(j)} y$ for two integers $i$ and $j$ with $1\le i< j\le n$, we have $x\succ_{\sigma(k)} y$ for every $i\le k\le j$. The following observation is immediate from the definition of single crossing profiles.

\begin{observation}
 Suppose a profile $\PP$ is single crossing with respect to an ordering $\sigma\in\SB_n$ of votes. Then \PP is single crossing with respect to the ordering $\overleftarrow{\sigma}$ too.
\end{observation}

\paragraph*{Problem Formulation} ~Suppose we have a profile \PP with $n$ voters and $m$ candidates. Let us define a function $\text{\Query}(x \succ_\el y)$ for a voter \el and two different candidates $x$ and $y$ to be \true if the voter \el prefers the candidate $x$ over the candidate $y$ and \false otherwise. We now formally define the problem.

\begin{definition}\PE\\
 Given an oracle access to \Query($\cdot$) for a single crossing profile \PP, find \PP.
\end{definition}

For two distinct candidates $x, y\in \CC$ and a voter \el, we say a \PE algorithm \AA {\em compares} candidates $x$ and $y$ for voter \el, if \AA makes a call to either $\text{\Query}(x \succ_\el y)$ or $\text{\Query}(y \succ_\el x)$. We define the number of queries made by the algorithm \AA, called the {\em query complexity} of \AA, to be the number of distinct tuples $(\el, x, y)\in \VV\times\CC\times\CC$ with $x\ne y$ such that the algorithm \AA compares the candidates $x$ and $y$ for voter \el. Notice that, even if the algorithm \AA makes multiple calls to \Query($\cdot$) with same tuple $(\el, x, y)$, we count it only once in the query complexity of \AA. This is without loss of generality since we can always implement a wrapper around the oracle which memorizes all the calls made to the oracle so far and whenever it receives a duplicate call, it replies from its memory without ``actually'' making a call to the oracle. We say two query complexities $\qqq(m,n)$ and $\qqq^\pr(m,n)$ are tight up to a factor of \el {\em for a large number of voters} if $\nfrac{1}{\el} \le \lim_{n\to\infty} \nfrac{\qqq(m,n)}{\qqq^\pr(m,n)} \le \el$.

Note that by using a standard sorting routine like merge sort, we can fully elicit an unknown preference using $\BigO(m \log m)$ queries. We state this explicitly below, as it will be useful in our subsequent discussions. 

\begin{observation}\label{obs:naive}
 There is a \PE algorithm for eliciting a single preference with query complexity $\BigO(m\log m)$.
\end{observation}

\paragraph*{Model of Input} ~We study two models of input for \PE for single crossing profiles.
\begin{itemize}
 \item {\bf Random access to voters:} In this model, we have a set of voters and we are allowed to ask any voter to compare any two candidates at any point of time. Moreover, we are also allowed to interleave the queries to different voters. Random access to voters is the model of input for elections within an organization where every voter belongs to the organization and can be queried any time.
 \item {\bf Sequential access to voters:} In this model, voters are arriving in a sequential manner one after another to the system. Once a voter \el arrives, we can query voter \el as many times as we like and then we ``release'' the voter \el from the system to grab the next voter in the queue. Once voter \el is released, it can never be queried again. Sequential access to voters is indeed the model of input in many practical elections scenarios such as political elections, restaurant ranking etc.
\end{itemize}

%% file: results.tex
\section{Results}
In this section, we present our technical results. \shortversion{In the interest of space, we skip proofs of a few results.} We first consider the (simpler) situation when the single crossing order is known, and then turn to the case when the order is unknown. In both cases, we explore all the relevant access models.

\subsection{Results: Known Single Crossing Order}

We begin with a simple \PE algorithm when we are given a random access to the voters and a single crossing ordering is known.

\begin{lemma}\label{lem:sc_random_ub}
 Suppose a profile \PP is single crossing with respect to a known permutation of the voters. Given a random access to voters, there is a \PE algorithm with query complexity $\BigO(m^2 \log n)$.
\end{lemma}

\begin{proof}
 By renaming, we assume that the profile is single crossing with respect to the identity permutation of the votes. Now, for every ${m\choose 2}$ pair of candidates $\{x, y\}\subset\CC$, we perform a binary search over the votes to find the index $i(\{x, y\})$ where the ordering of $x$ and $y$ changes. We now know how any voter $j$ orders any two candidates $x$ and $y$ from $i(\{x, y\})$ and thus we have found $\mathcal{P}$.
\end{proof}
Interestingly, the simple algorithm in \Cref{lem:sc_random_ub} turns out to be optimal up to a multiplicative factor of $\OO(m)$ as we prove next. The idea is to ``pair up'' the candidates and design an oracle which ``hides'' the vote where the ordering of the two candidates in any pair $(x, y)$ changes unless it receives at least $(\log m - 1)$ queries involving only these two candidates $x$ and $y$. We formalize this idea below.

\begin{theorem}\label{thm:sc_random_lb}
 Suppose a profile \PP is single crossing with respect to the identity permutation of votes. Given random access to voters, any \PE algorithm has query complexity $\Omega(m\log m + m\log n)$.
\end{theorem}

\begin{proof}
 The $\Omega(m\log m)$ bound follows from the query complexity lower bound of sorting and the fact that any profile consisting of only one preference $\succ\in\LL(\CC)$ is single crossing. Let $\CC = \{c_1, \ldots, c_m\}$ be the set of $m$ candidates where $m$ is an even integer. Consider the ordering $Q = c_1 \succ c_2 \succ \cdots \succ c_m \in\LL(\CC)$ and the following pairing of the candidates: $\{c_1, c_2\}, \{c_3, c_4\}, \ldots, \{c_{m-1}, c_m\}$. Our oracle answers \Query($\cdot$) as follows. The oracle fixes the preferences of the voters one and $n$ to be $Q$ and $\overleftarrow{Q}$ respectively. For every odd integer $i\in[m]$, the oracle maintains $\theta_i$ (respectively $\beta_i$) which corresponds to the largest (respectively smallest) index of the voter for whom $(c_i, c_{i+1})$ has already been queried and the oracle answered that the voter prefers $c_i$ over $c_{i+1}$ ($c_{i+1}$ over $c_i$ respectively). The oracle initially sets $\theta_i = 1$ and $\beta_i = n$ for every odd integer $i\in[m]$. Suppose oracle receives a query to compare candidates $c_i$ and $c_j$ for $i,j\in[m]$ with $i<j$ for a voter $\ell$. If $i$ is an even integer or $j-i\ge 2$ (that is, $c_i$ and $c_j$ belong to different pairs), then the oracle answers that the voter \el prefers $c_i$ over $c_j$. Otherwise we have $j=i+1$ and $i$ is an odd integer. The oracle answers the query to be $c_i \succ c_{i+1}$ and updates $\theta_i$ to \el keeping $\beta_i$ fixed if $|\el-\theta_i| \le |\el-\beta_i|$ and otherwise answers $c_{i+1} \succ c_i$ and updates $\beta_i$ to \el keeping $\theta_i$ fixed (that is, the oracle answers according to the vote which is closer to the voter \el between $\theta_i$ and $\beta_i$ and updates $\theta_i$ or $\beta_i$ accordingly). If the pair $(c_i, c_{i+1})$ is queried less than $(\log n - 2)$ times, then we have $\beta_i - \theta_i \ge 2$ at the end of the algorithm since every query for the pair $(c_i, c_{i+1})$ decreases $\beta_i - \theta_i$ by at most a factor of two and we started with $\beta_i - \theta_i = n-1$. Consider a voter $\kappa$ with $\theta_i < \kappa < \beta_i$. If the elicitation algorithm outputs that the voter $\kappa$ prefers $c_i$ over $c_{i+1}$ (respectively $c_{i+1}$ over $c_i$), then the oracle sets all the voters $\kappa^\pr$ with $\theta_i < \kappa^\pr < \beta_i$ to prefer $c_{i+1}$ over $c_i$ (respectively $c_i$ over $c_{i+1}$). Clearly, the algorithm does not elicit the preference of the voter $\kappa$ correctly. Also, the profile is single crossing with respect to the identity permutation of the voters and consistent with the answers of all the queries made by the algorithm. Hence, for every odd integer $i\in[m]$, the algorithm must make at least $(\log n - 1)$ queries for the pair $(c_i, c_{i+1})$ thereby making $\Omega(m\log n)$ queries in total.
\end{proof}

We now present our \PE algorithm when we have a sequential access to the voters according to a single crossing order. We elicit the preference of the first voter using \Cref{obs:naive}. From second vote onwards, we simply use the idea of {\em insertion sort} relative to the previously elicited vote \cite{cormen2009introduction}. Since we are using insertion sort, any particular voter may be queried $\OO(m^2)$ times. However, we are able to bound the query complexity of our algorithm due to two fundamental reasons: (i) consecutive preferences will often be almost similar in a single crossing ordering, (ii) our algorithm takes only $\BigO(m)$ queries to elicit the preference of the current voter if its preference is indeed the same as the preference of the voter preceding it. In other words, every time we have to ``pay'' for shifting a candidate further back in the current vote, the relative ordering of that candidate with all the candidates that it jumped over is now fixed, because for these pairs, the one permitted crossing is now used up. 
We begin with presenting an important subroutine called Elicit($\cdot$) which finds the preference of a voter \el given another preference \RR by performing an insertion sort using \RR as the order of insertion.

\begin{algorithm}[ht]
 \caption{Elicit(\CC, \RR, \el)}\label{alg:elicit}
\begin{algorithmic}[1]
 \Require{A set of candidates $\CC = \{c_i : i\in[m]\}$, an ordering $\RR = c_1 \succ \cdots \succ c_m$ of \CC, a voter \el}
 \Ensure{Preference ordering $\succ_\el$ of voter \el on \CC}
 \State $\QQ\leftarrow c_1$ \Comment{\QQ will be the preference of the voter \el}
 \For{$i \gets 2 \textrm{ to } m$}\label{elicit_for}\Comment{$c_i$ is inserted in the $i^{th}$ iteration}
  \State Scan \QQ linearly {\em from index $i-1$ to $1$} to find the index $j$ where $c_i$ should be inserted according to the preference of voter \el and insert $c_i$ in \QQ at $j$\label{elicit_scan}
 \EndFor
 \Return \QQ
\end{algorithmic}
\end{algorithm}

For the sake of the analysis of our algorithm, let us to introduce a few terminologies. Given two preferences $\succ_1$ and $\succ_2$, we call a pair of candidates $(x, y)\in\CC\times\CC, x\ne y,$ {\em good} if both $\succ_1$ and $\succ_2$ order them in a same way; a pair of candidates is called {\em bad} if it is not good. We divide the number of queries made by our algorithm into two parts: goodCost($\cdot$) and badCost($\cdot$) which are the number of queries made between good and respectively bad pair of candidates. In what follows, we show that goodCost($\cdot$) for Elicit($\cdot$) is small and the total badCost($\cdot$) across all the runs of Elicit($\cdot$) is small. 

\begin{lemma}\label{lem:good}
 The goodCost(Elicit($\CC, \RR, \el$)) of Elicit($\CC, \RR, \el$) is $\OO(m)$ (good is with respect to the preferences \RR and $\succ_\el$).
\end{lemma}
\begin{proof}
 Follows immediately from the observation that in any iteration of the for loop at line \ref{elicit_for} in \Cref{alg:elicit}, only one good pair of candidates are compared.
\end{proof}

We now use \Cref{alg:elicit} iteratively to find the profile. We present the pseudocode in \Cref{alg:final} which works for the more general setting where a single crossing ordering is known but the voters are arriving in any arbitrary order $\pi$. We next compute the query complexity of \Cref{alg:final} when voters are arriving in a single crossing order.

\begin{algorithm}[ht]
\caption{PreferenceElicit($\pi$)}\label{alg:final}
 \begin{algorithmic}[1]
  \Require{$\pi\in\SB_n$}
  \Ensure{Profile of all the voters}
  \State $\QQ[\pi(1)] \leftarrow$ Elicit $\succ_{\pi(1)}$ using \Cref{obs:naive} \Comment{\QQ stores the profile}
  \State $\XX \leftarrow \{\pi(1)\}$ \Comment{Set of voters' whose preferences have already been elicited}
  \For{$i \gets 2 \textrm{ to } n$} \Comment{Elicit the preference of voter $\pi(i)$ in iteration $i$}
   \State $k \leftarrow \min_{j\in\XX} |j-i|$\label{final_k} \Comment{Find the closest known preference}
   \State $\RR \leftarrow \QQ[k], \XX \leftarrow \XX\cup\{\pi(i)\}$
   \State $\QQ[\pi(i)] \leftarrow \text{ Elicit}(\CC, \RR, \pi(i))$
  \EndFor
  \Return \QQ
 \end{algorithmic}
\end{algorithm}

\begin{theorem}\label{thm:sc_seq_known_ub}
 Assume that the voters are arriving sequentially according to an order with respect to which a profile \PP is single crossing. Then there is a \PE algorithm with query complexity $\BigO(mn + m^2)$.
\end{theorem}

\begin{proof}
 By renaming, let us assume, without loss of generality, that the voters are arriving according to the identity permutation $id_n$ of the voters and the profile \PP is single crossing with respect to $id_n$. Let the profile $\PP$ be $(P_1, P_2, \ldots, P_n) \in \LL(\CC)^n$. For two candidates $x, y\in \CC$ and a voter $i\in\{2, \ldots, n\}$, let us define a variable $b(x,y,i)$ to be one if $x$ and $y$ are compared for the voter $i$ by Elicit(\CC,$P_{i-1}$, $i$) and $(x, y)$ is a bad pair of candidates with respect to the preferences of voter $i$ and $i-1$; otherwise $b(x,y,i)$ is defined to be zero. Then we have the following. 
 \begin{eqnarray*}
  \shortversion{&&}\text{CostPreferenceElicit}(id_n)\shortversion{\\} &=& \BigO(m\log m) +\sum_{i=2}^n\mathlarger{\mathlarger{(}}\text{goodCost(\Query}(\CC, P_{i-1}, i)) + \shortversion{\\&&}\text{badCost(\Query}(\CC, P_{i-1}, i))\mathlarger{\mathlarger{)}}\\
  &\le& \BigO(m\log m + mn) + \sum_{i=2}^n\text{badCost(\Query}(\CC, P_{i-1}, i))\\
  &=& \BigO(m\log m + mn) + \mathlarger{\sum}_{(x, y)\in\CC\times\CC} \left(\sum_{i=2}^n b(x, y, i)\right)\\
  &\le& \BigO(m\log m + mn) + \sum_{(x, y)\in\CC\times\CC} 1\\
  &=& \BigO(mn + m^2)
 \end{eqnarray*} 
 The first inequality follows from \Cref{lem:good}, the second equality follows from the definition of $b(x, y, i)$, and the second inequality follows from the fact that $\sum_{i=2}^n b(x, y, i) \le 1$ for every pair of candidates $(x,y)\in\CC$ since the profile \PP is single crossing.
\end{proof}

We show next that, when the voters are arriving in a single crossing order, the query complexity upper bound in \Cref{thm:sc_seq_known_lb} is tight for a large number of voters up to constant factors. The idea is to pair up the candidates in a certain way and argue that the algorithm must compare the candidates in every pair for every voter thereby proving a $\Omega(mn)$ lower bound on query complexity.

\begin{theorem}\label{thm:sc_seq_known_lb}
 Assume that the voters are arriving sequentially according to an order with respect to which a profile \PP is single crossing. Then any \PE algorithm has query complexity $\Omega(m\log m + mn)$.
\end{theorem}

\begin{proof}
 The $\Omega(m\log m)$ bound follows from the fact that any profile consisting of only one preference $P\in\LL(\CC)$ is single crossing. By renaming, let us assume without loss of generality that the profile \PP is single crossing with respect to the identity permutation of the voters. Suppose we have an even number of candidates and $\CC = \{c_1, \ldots, c_m\}$. Consider the order $\QQ = c_1 \succ c_2 \succ \cdots \succ c_m$ and the pairing of the candidates $\{c_1, c_2\}, \{c_3, c_4\}, \ldots, \{c_{m-1}, c_m\}$. The oracle answers all the query requests consistently according to the order $Q$ till the first voter $\kappa$ for which there exists at least one odd integer $i\in[m]$ such that the pair $(c_i, c_{i+1})$ is not queried. If there does not exist any such $\kappa$, then the algorithm makes at least $\nfrac{mn}{2}$ queries thereby proving the statement. Otherwise, let $\kappa$ be the first vote such that the algorithm does not compare $c_i$ and $c_{i+1}$ for some odd integer $i\in[m]$. The oracle answers the queries for the rest of the voters $\{\kappa+1, \ldots, n\}$ according to the order $Q^\pr = c_1 \succ c_2 \succ \cdots \succ c_{i-1} \succ c_{i+1} \succ c_i \succ c_{i+2} \succ \cdots \succ c_m$. If the algorithm orders $c_i \succ_\kappa c_{i+1}$ in the preference of the voter $\kappa$, then the oracle sets the preference of the voter $\kappa$ to be $\QQ^\pr$. On the other hand, if the algorithm orders $c_{i+1} \succ_\kappa c_i$ in the preference of voter $\kappa$, then the oracle sets the preference of voter $\kappa$ to be \QQ. Clearly, the elicitation algorithm fails to correctly elicit the preference of the voter $\kappa$. However, the profiles for both the cases are single crossing with respect to the identity permutation of the voters and are consistent with the answers given to all the queries made by the algorithm. Hence, the algorithm must make at least $\nfrac{mn}{2}$ queries.
\end{proof}

We next move on to the case when we know a single crossing order of the voters; however, the voters arrive in an arbitrary order $\pi\in\SB_n$. The idea is to call the function Elicit($\CC, \RR, i$) where the current voter is the voter $i$ and \RR is the preference of the voter which is closest to $i$ according to a single crossing ordering and whose preference has already been elicited by the algorithm.

\begin{theorem}\label{thm:sc_seq_any_ub}
 Assume that a profile \PP is known to be single crossing with respect to a known ordering of voters $\sigma\in\SB_n$. However, the voters are arriving sequentially according to an arbitrary order $\pi\in\SB_n$ which may be different from $\sigma$. Then there is a \PE algorithm with query complexity $\BigO(mn + m^2\log n)$.
\end{theorem}

\begin{proof}
 By renaming, let us assume, without loss of generality, that the profile \PP is single peaked with respect to the identity permutation of the voters. Let the profile $\PP$ be $(P_1, P_2, \ldots, P_n) \in \LL(\CC)^n$. Let $f:[n]\longrightarrow [n]$ be the function such that $f(i)$ is the $k$ corresponding to the $i$ at line \ref{final_k} in \Cref{alg:final}. For candidates $x, y\in\CC$ and voter $\ell$, we define $b(x, y, \el)$ analogously as in the proof of \Cref{thm:sc_seq_known_ub}. We claim that $ B(x, y) = \sum_{i=2}^n b(x, y, i) \le \log n$. To see this, we consider any arbitrary pair $(x, y)\in\CC\times\CC$. Let the set of indices of the voters that have arrived immediately after the first time $(x,y)$ contributes to $B(x, y)$ be $\{i_1, i_2, \ldots, i_t\}$. Without loss of generality, let us assume $i_1 < i_2 < \cdots < i_t$. Again, without loss of generality, let us assume that voters $i_1, i_2, \ldots, i_j$ prefer $x$ over $y$ and voters $i_{j+1}, \ldots, i_t$ prefer $y$ over $x$. Let us define $\Delta$ to be the difference between smallest index of the voter who prefers $y$ over $x$ and the largest index of the voter who prefers $x$ over $y$. Hence, we currently have $\Delta = i_{j+1} - i_j$. A crucial observation is that if a new voter \el contributes to $B(x, y)$ then we must necessarily have $i_j < \el < i_{j+1}$. Another crucial observation is that whenever a new voter contributes to $B(x, y)$, the value of $\Delta$ gets reduced at least by a factor of two by the choice of $k$ at line \ref{final_k} in \Cref{alg:final}. Hence, the pair $(x,y)$ can contribute at most $(1+\log \Delta) = \BigO(\log n)$ to $B(x, y)$ since we have $\Delta\le n$ to begin with.\shortversion{ The rest of the proof is along the same line of the proof of \Cref{thm:sc_seq_known_ub}.}\longversion{ Then we have the following. 
 \sloppypar
 \begin{eqnarray*}
  \shortversion{&&}\text{CostPreferenceElicit}(\pi) \shortversion{\\} &=& \BigO(m\log m) +\sum_{i=2}^n\text{goodCost(\Query}(\CC, P_{f(i)}, i)) + \shortversion{\\&&} \text{badCost(\Query}(\CC, P_{f(i)}, i))\\
  &\le& \BigO(m\log m + mn) + \sum_{i=2}^n\text{badCost(\Query}(\CC, P_{f(i)}, i))\\
  &=& \BigO(m\log m + mn) + \sum_{(x, y)\in\CC\times\CC} \sum_{i=2}^n b(x, y, i)\\
  &\le& \BigO(m\log m + mn) + \sum_{(x, y)\in\CC\times\CC} \log n\\
  &=& \BigO(mn + m^2\log n)
 \end{eqnarray*} 
 The first inequality follows from \Cref{lem:good}, the second equality follows from the definition of $b(x, y, i)$, and the second inequality follows from the fact that $\sum_{i=2}^n b(x, y, i) \le \log n$.}
\end{proof}

\subsection{Results: Unknown Single Crossing Order}

We now turn our attention to \PE for single crossing profiles when no single crossing ordering is known. Before we present our \PE algorithm for this setting, let us first prove a few structural results about single crossing profiles which we will use crucially later. We begin with showing an upper bound on the number of distinct preferences in any single crossing profile.

\begin{lemma}\label{lem:bound}
 Let \PP be a profile on a set \CC of candidates which is single crossing. Then the number of distinct preferences in \PP is at most ${m\choose 2}$.
\end{lemma}

\begin{proof}
 By renaming, let us assume, without loss of generality, that the profile \PP is single crossing with respect to the identity permutation of the voters. We now observe that whenever the $i^{th}$ vote is different from the $(i+1)^{th}$ vote for some $i\in[n-1]$, there must exist a pair of candidates $(x, y)\in \CC\times\CC$ whom the $i^{th}$ vote and the $(i+1)^{th}$ vote order differently. Now the statement follows from the fact that, for every pair of candidates $(a,b)\in\CC\times\CC$, there can exist at most one $i\in[n-1]$ such that the $i^{th}$ vote and the $(i+1)^{th}$ vote order $a$ and $b$ differently.
\end{proof}

We show next that in every single crossing profile \PP where all the preferences are {\em distinct}, there exists a pair of candidates $(x, y)\in\CC\times\CC$ such that nearly half of the voters in \PP prefer $x$ over $y$ and the other voters prefer $y$ over $x$.

\begin{lemma}\label{lem:divide}
 Let \PP be a profile of $n$ voters such that all the preferences are distinct. Then there exists a pair of candidates $(x, y)\in\CC$ such that $x$ is preferred over $y$ in at least $\lfloor\nfrac{n}{2}\rfloor$ preferences and $y$ is preferred over $x$ in at least $\lfloor\nfrac{n}{2}\rfloor$ preferences in \PP.
\end{lemma}

\begin{proof}
 Without loss of generality, by renaming, let us assume that the profile \PP is single crossing with respect to the identity permutation of the voters. Since all the preferences in \PP are distinct, there exists a pair of candidates $(x, y)\in\CC\times\CC$ such that the voter $\lfloor\nfrac{n}{2}\rfloor$ and the voter $\lfloor\nfrac{n}{2}\rfloor + 1$ order $x$ and $y$ differently. Let us assume, without loss of generality, that the voter $\lfloor\nfrac{n}{2}\rfloor$ prefers $x$ over $y$. Now, since the profile \PP is single crossing, every voter in $[\lfloor\nfrac{n}{2}\rfloor]$ prefer $x$ over $y$ and every voter in $\{\lfloor\nfrac{n}{2}\rfloor + 1, \ldots, n\}$ prefer $y$ over $x$.
\end{proof}

Using \Cref{lem:bound,lem:divide} we now design a \PE algorithm when no single crossing ordering of the voters is known. The overview of the algorithm is as follows. At any point of time in the elicitation process, we have the set \QQ of all the distinct preferences that we have already elicited completely and we have to elicit the preference of a voter $\ell$. We first search the set of votes \QQ for a preference which is {\em possibly} same as the preference $\succ_\el$ of the voter $\ell$. It turns out that we can find a possible match $\succ\in\QQ$ using $\OO(\log|\QQ|)$ queries due to \Cref{lem:divide} which is $\OO(\log m)$ due to \Cref{lem:bound}. We then check whether the preference of the voter \el is indeed the same as $\succ$ or not using $\OO(m)$ queries. If $\succ$ is the same as $\succ_\el$, then we have elicited $\succ_\el$ using $\OO(m)$ queries. Otherwise, we elicit $\succ_\el$ using $\OO(m\log m)$ queries using \Cref{obs:naive}. Fortunately, \Cref{lem:bound} tells us that we would use the algorithm in \Cref{obs:naive} at most $\OO(m^2)$ times. We present the pseudocode of our \PE algorithm in this setting in \Cref{alg:unknown}. \longversion{It uses \Cref{alg:same} as a subroutine which returns \true if the preference of any input voter is same as any given preference.

\begin{algorithm}[ht]
 \caption{Same(\RR, \el)}\label{alg:same}
 \begin{algorithmic}[1]
  \Require{$\RR = c_1\succ c_2\succ \cdots \succ c_m \in\LL(\CC) ,\el\in[n]$}
  \Ensure{\true if the preference of the $\el^{th}$ voter is \RR; \false otherwise}
  \For{$i \gets 1 \textrm{ to } m-1$}
   \If{\Query($c_i \succ_\el c_{i+1}$) = \false}
    \Return \false \Comment{We have found a mismatch.}
   \EndIf
  \EndFor
  \Return \true
 \end{algorithmic} 
\end{algorithm}
}

\begin{algorithm}[ht]
\caption{\longversion{PreferenceElicitUnknownSingleCrossingOrdering}\shortversion{PreferenceElicitUnknownSCOrdering}($\pi$)}\label{alg:unknown}
 \begin{algorithmic}[1]
  \Require{$\pi\in\SB_n$}
  \Ensure{Profile of all the voters}
  \State $\RR, \QQ \leftarrow \emptyset$ \Comment{\QQ stores all the votes seen so far without duplicate. \RR stores the profile.}
  \For{$i \gets 1 \textrm{ to } n$} \Comment{Elicit preference of the $i^{th}$ voter in $i^{th}$ iteration of this for loop.}
   \State $\QQ^\pr \leftarrow \QQ$ 
   \While{$|\QQ^\pr|>1$} \Comment{Search \QQ to find a vote potentially same as the preference of $\pi(i)$}
    \State Let $x, y\in \CC$ be two candidates such that at least $\lfloor\nfrac{|\QQ^\pr|}{2}\rfloor$ votes in $\QQ^\pr$ prefer $x$ over $y$ and at least $\lfloor\nfrac{|\QQ^\pr|}{2}\rfloor$ votes in $\QQ^\pr$ prefer $y$ over $x$.
    \If{\Query($x \succ_{\pi(i)} y$) = \true}
     \State $\QQ^\pr \leftarrow \{ v\in\QQ^\pr : v \text{ prefers } x \text{ over } y \}$
    \Else
     \State $\QQ^\pr \leftarrow \{ v\in\QQ^\pr : v \text{ prefers } y \text{ oer } x \}$
    \EndIf
   \EndWhile
   \State Let $w$ be the only vote in $\QQ^\pr$\label{potential_match} \Comment{$w$ is potentially same as the preference of $\pi(i)$}
   \If{Same($w, \pi(i)$) = \true} \Comment{Check whether the vote $\pi(i)$ is potentially same as $w$}
    \State $\RR[\pi(i)] \leftarrow w$
   \Else 
    \State $\RR[\pi(i)] \leftarrow$ Elicit using \Cref{obs:naive}
    \State $\QQ \leftarrow \QQ\cup\{\RR[\pi(i)]\}$
   \EndIf
  \EndFor
  \Return \RR
 \end{algorithmic}
\end{algorithm}

\begin{theorem}\label{thm:sc_seq_unknown_ub}
 Assume that a profile \PP is known to be single crossing. However, no ordering of the voters with respect to which \PP is single crossing is known. The voters are arriving sequentially according to an arbitrary order $\pi\in\SB_n$. Then there is a \PE algorithm with query complexity $\BigO(mn + m^3\log m)$.
\end{theorem}

\begin{proof}
 We present the pesudocode in \Cref{alg:unknown}. We maintain two arrays in the algorithm. The array \RR is of length $n$ and the $j^{th}$ entry stores the preference of voter $j$. The other array \QQ stores all the votes seen so far after removing duplicate votes; more specifically, if some specific preference $\suc$ has been seen \el many times for any $\el>0$, \QQ stores only one copy of $\suc$. Upon arrival of voter $i$, we first check whether there is a preference in \QQ which is ``potentially'' same as the preference of voter $i$. At the beginning of the search, our search space $\QQ^\pr=\QQ$ for a potential match in \QQ is of size $|\QQ|$. We next iteratively keep halving the search space as follows. We find a pair of candidates $(x, y)\in\CC\times\CC$ such that at least $\lfloor\nfrac{|\QQ^\pr|}{2}\rfloor$ preferences in $\QQ^\pr$ prefer $x$ over $y$ and at least $\lfloor\nfrac{|\QQ^\pr|}{2}\rfloor$ preferences prefer $y$ over $x$. The existence of such a pair of candidates is guaranteed by \Cref{lem:divide} and can be found in $\OO(m^2)$ time by simply going over all possible pairs of candidates. By querying how voter $i$ orders $x$ and $y$, we reduce the search space $\QQ^\pr$ for a potential match in \QQ to a set of size at most $\lfloor\nfrac{|\QQ^\pr|}{2}\rfloor+1$. Hence, in $\BigO(\log m)$ queries, the search space reduces to only one preference since we have $|\QQ|\le m^2$ by \Cref{lem:bound}. Once we find a potential match $w$ in \QQ (line \ref{potential_match} in \Cref{alg:unknown}), we check whether the preference of voter $i$ is the same as $w$ or not using $\BigO(m)$ queries. If the preference of voter $i$ is indeed same as $w$, then we output $w$ as the preference of voter $i$. Otherwise, we use \Cref{obs:naive} to elicit the preference of voter $i$ using $\BigO(m\log m)$ queries and put the preference of voter $i$ in \QQ. Since the number of times we need to use the algorithm in \Cref{obs:naive} is at most the number of distinct votes in \PP which is known to be at most $m^2$ by \Cref{lem:bound}, we get the statement.
\end{proof}
\Cref{thm:sc_seq_unknown_ub} immediately gives us the following corollary in the random access to voters model when no single crossing ordering is known.

\begin{corollary}\label{cor:sc_random_unknown_ub}
 Assume that a profile \PP is known to be single crossing. However, no ordering of the voters with respect to which \PP is single crossing is known. Given a random access to voters, there is an \PE algorithm with query complexity $\BigO(mn + m^3\log m)$.
\end{corollary}

\longversion{
\begin{proof}
 \Cref{alg:unknown} works for this setting also and exact same bound on the query complexity holds.
\end{proof}
}
We now show that the query complexity upper bound of \Cref{cor:sc_random_unknown_ub} is tight up to constant factors for large number of voters.

\begin{theorem}\label{thm:sc_random_unknown_lb}
 Given a random access to voters, any \PE algorithm which do not know any ordering of the voters with respect to which the input profile is single crossing has query complexity $\Omega(m\log m + mn)$.
\end{theorem}

\begin{proof}
 The $\Omega(m\log m)$ bound follows from sorting lower bound and the fact that any profile consisting of only one preference $P\in\LL(\CC)$ is single crossing. Suppose we have an even number of candidates and $\CC = \{c_1, \ldots, c_m\}$. Consider the ordering $\QQ = c_1 \succ c_2 \succ \cdots \succ c_m$ and the pairing of the candidates $\{c_1, c_2\}, \{c_3, c_4\}, \ldots, \{c_{m-1}, c_m\}$. The oracle answers all the query requests consistently according to the ordering $Q$. We claim that any \PE algorithm \AA must compare $c_i$ and $c_{i+1}$ for every voter and for every odd integer $i\in[m]$. Indeed, otherwise, there exist a voter $\kappa$ and an odd integer $i\in[m]$ such that the algorithm \AA does not compare $c_i$ and $c_{i+1}$. Suppose the algorithm outputs a profile $\PP^\pr$. If the voter $\kappa$ prefers $c_i$ over $c_{i+1}$ in $\PP^\pr$, then the oracle fixes the preference $\suc_\kappa$ to be $c_1 \succ c_2 \succ \cdots \succ c_{i-1} \succ c_{i+1} \succ c_i \succ c_{i+2} \succ \cdots \succ c_m$; otherwise the oracle fixes $\suc_\kappa$ to be \QQ. The algorithm fails to correctly output the preference of the voter $\kappa$ in both the cases. Also the final profile with the oracle is single crossing with respect to any ordering of the voters that places the voter $\kappa$ at the end. Hence, \AA must compare $c_i$ and $c_{i+1}$ for every voter and for every odd integer $i\in[m]$ and thus has query complexity $\Omega(mn)$.
\end{proof}

\longversion{
\subsection{Nearly Single Crossing}

We now consider preference elicitation for profiles which are nearly single crossing. We begin with profiles with bounded single crossing width.

\begin{proposition}
 Suppose a profile \PP is single crossing with width $w$. Given a \PE algorithm \AA with query complexity $\qqq(m,n)$ for random or sequential access to the voters when a single crossing order is known or unknown, there exists another \PE algorithm $\AA^\pr$ for the single crossing profiles with width $w$ which has query complexity $\OO(\qqq(\nfrac{m}{w},n) + mn\log w)$ under same setting as \AA.
\end{proposition}

\begin{proof}
 Let the partition of the set of candidates \CC with respect to which the profile \PP is single crossing be $\bCC_i, i\in[\lceil\nfrac{m}{w}\rceil]$. Hence, $\CC = \cup_{i\in[\lceil\nfrac{m}{w}\rceil]} \bCC_i$ and $\bCC_i \cap \bCC_j = \emptyset$ for every $i, j\in[\lceil\nfrac{m}{w}\rceil]$ with $i\ne j$. Let $\CC^\pr$ be a subset of candidates containing exactly one candidate from $\bCC_i$ for each $i\in[\lceil\nfrac{m}{w}\rceil]$. We first find $\PP(\CC^\pr)$ using \AA. The query complexity of this step is $\OO(\qqq(\nfrac{m}{w},n))$. Next we find $\PP(\bCC_i)$ using \Cref{obs:naive} for every $i\in[\nfrac{m}{w}]$ thereby finding \PP. The overall query complexity is $\OO(\qqq(\nfrac{m}{w},n) + (\nfrac{m}{w})nw\log w) = \OO(\qqq(\nfrac{m}{w},n) + mn\log w)$.
\end{proof}
}

%% file: conclusion.tex
\shortversion{\vspace{-2ex}}

\section{Conclusions and Future Work}

In this paper, we have presented \PE algorithms with low query complexity for single crossing profiles under various settings. Moreover, we have proved that the query complexity of our algorithms are tight for large number of voters up to constant factors for all but one setting namely when the voters can be accessed randomly but we do not know any ordering with respect to which the voters are single crossing.\longversion{ An immediate future work is to have a matching query complexity bound for the only open case.} We do not assume any prior knowledge about the profile. It would be interesting to study the effect of prior knowledge on the query complexity for \longversion{preference} elicitation.

%% file: Elicitation_Single_Crossing.bbl
\newcommand{\etalchar}[1]{$^{#1}$}
\begin{thebibliography}{BNM{\etalchar{+}}58}

\bibitem[Arr50]{arrow1950difficulty}
Kenneth~J Arrow.
\newblock A difficulty in the concept of social welfare.
\newblock {\em The Journal of Political Economy}, pages 328--346, 1950.

\bibitem[BBHH15]{brandt2015bypassing}
Felix Brandt, Markus Brill, Edith Hemaspaandra, and Lane~A Hemaspaandra.
\newblock Bypassing combinatorial protections: Polynomial-time algorithms for
  single-peaked electorates.
\newblock {\em Journal of Artificial Intelligence Research (JAIR)}, pages
  439--496, 2015.

\bibitem[BCW13]{bredereck2013characterization}
Robert Bredereck, Jiehua Chen, and Gerhard~J Woeginger.
\newblock A characterization of the single-crossing domain.
\newblock {\em Social Choice and Welfare}, 41(4):989--998, 2013.

\bibitem[BH11]{ballester2011characterization}
Miguel~A Ballester and Guillaume Haeringer.
\newblock A characterization of the single-peaked domain.
\newblock {\em Social Choice and Welfare}, 36(2):305--322, 2011.

\bibitem[BITT89]{bartholdi1989voting}
John Bartholdi~III, Craig~A Tovey, and Michael~A Trick.
\newblock Voting schemes for which it can be difficult to tell who won the
  election.
\newblock {\em Social Choice and welfare}, 6(2):157--165, 1989.

\bibitem[Bla48]{black1948rationale}
Duncan Black.
\newblock On the rationale of group decision-making.
\newblock {\em The Journal of Political Economy}, pages 23--34, 1948.

\bibitem[BM11]{barbera2011top}
Salvador Barber{\`a} and Bernardo Moreno.
\newblock Top monotonicity: A common root for single peakedness, single
  crossing and the median voter result.
\newblock {\em Games and Economic Behavior}, 73(2):345--359, 2011.

\bibitem[BNM{\etalchar{+}}58]{black1958theory}
Duncan Black, Robert~Albert Newing, Iain McLean, Alistair McMillan, and Burt~L
  Monroe.
\newblock {\em The theory of committees and elections}.
\newblock Springer, 1958.

\bibitem[CC83]{chamberlin1983representative}
John~R Chamberlin and Paul~N Courant.
\newblock Representative deliberations and representative decisions:
  Proportional representation and the borda rule.
\newblock {\em American Political Science Review}, 77(03):718--733, 1983.

\bibitem[CGS13]{cornaz2013kemeny}
Denis Cornaz, Lucie Galand, and Olivier Spanjaard.
\newblock Kemeny elections with bounded single-peaked or single-crossing width.
\newblock In {\em Proceedings of the Twenty-Third International Joint
  Conference on Artificial Intelligence (IJCAI)}, pages 76--82. AAAI Press,
  2013.

\bibitem[Con09]{Conitzer09}
Vincent Conitzer.
\newblock Eliciting single-peaked preferences using comparison queries.
\newblock {\em J. Artif. Intell. Res. {(JAIR)}}, 35:161--191, 2009.

\bibitem[Cor09]{cormen2009introduction}
Thomas~H Cormen.
\newblock {\em Introduction to algorithms}.
\newblock MIT press, 2009.

\bibitem[CS02]{conitzer2002vote}
Vincent Conitzer and Tuomas Sandholm.
\newblock Vote elicitation: Complexity and strategy-proofness.
\newblock In {\em Eighteenth National Conference on Artificial Intelligence
  (AAAI)}, pages 392--397, 2002.

\bibitem[CS05]{conitzer2005communication}
Vincent Conitzer and Tuomas Sandholm.
\newblock Communication complexity of common voting rules.
\newblock In {\em Proceedings of the 6th ACM conference on Electronic Commerce
  (EC)}, pages 78--87. ACM, 2005.

\bibitem[DKNS01]{dwork2001rank}
Cynthia Dwork, Ravi Kumar, Moni Naor, and Dandapani Sivakumar.
\newblock Rank aggregation methods for the web.
\newblock In {\em Proceedings of the 10th International conference on World
  Wide Web (WWW)}, pages 613--622. ACM, 2001.

\bibitem[DL13]{ding2013voting}
Ning Ding and Fangzhen Lin.
\newblock Voting with partial information: what questions to ask?
\newblock In {\em Proceedings of the 12th International Conference on
  Autonomous Agents and Multi-agent Systems (AAMAS)}, pages 1237--1238.
  International Foundation for Autonomous Agents and Multiagent Systems, 2013.

\bibitem[DM16]{deypeak}
Palash Dey and Neeldhara Misra.
\newblock Elicitation for preferences single peaked on trees.
\newblock In {\em Proceedings of the Twenty-Fifth International Joint
  Conference on Artificial Intelligence (IJCAI)}. AAAI Press, 2016.

\bibitem[Dod76]{dodgson1876method}
Charles~Lutwidge Dodgson.
\newblock {\em A method of taking votes on more than two issues}.
\newblock 1876.

\bibitem[EFLO15]{ElkindFLO15}
Edith Elkind, Piotr Faliszewski, Martin Lackner, and Svetlana Obraztsova.
\newblock The complexity of recognizing incomplete single-crossing preferences.
\newblock In {\em Proceedings of the Twenty-Ninth {AAAI} Conference on
  Artificial Intelligence, January 25-30, 2015, Austin, Texas, {USA.}}, pages
  865--871, 2015.

\bibitem[EFS14]{elkind2014characterization}
Edith Elkind, Piotr Faliszewski, and Piotr Skowron.
\newblock A characterization of the single-peaked single-crossing domain.
\newblock In {\em Proceedings of the 28th Conference on Artificial Intelligence
  (AAAI)}, page~12, 2014.

\bibitem[EL{\"O}08]{escoffier2008single}
Bruno Escoffier, J{\'e}r{\^o}me Lang, and Meltem {\"O}zt{\"u}rk.
\newblock Single-peaked consistency and its complexity.
\newblock In {\em Proceedings of the 18th European Conference on Artificial
  Intelligence (ECAI)}, volume~8, pages 366--370, 2008.

\bibitem[FHH14]{faliszewski2014complexity}
Piotr Faliszewski, Edith Hemaspaandra, and Lane~A Hemaspaandra.
\newblock The complexity of manipulative attacks in nearly single-peaked
  electorates.
\newblock {\em Artificial Intelligence}, 207:69--99, 2014.

\bibitem[FHHR09]{faliszewski2009shield}
Piotr Faliszewski, Edith Hemaspaandra, Lane~A Hemaspaandra, and J{\"o}rg Rothe.
\newblock The shield that never was: Societies with single-peaked preferences
  are more open to manipulation and control.
\newblock In {\em Proceedings of the 12th Conference on Theoretical Aspects of
  Rationality and Knowledge}, pages 118--127. ACM, 2009.

\bibitem[Gib73]{gibbard1973manipulation}
Allan Gibbard.
\newblock Manipulation of voting schemes: a general result.
\newblock {\em Econometrica: Journal of the Econometric Society}, pages
  587--601, 1973.

\bibitem[HM97]{hinich1997analytical}
Melvin~J Hinich and Michael~C Munger.
\newblock {\em Analytical politics}.
\newblock Cambridge University Press, 1997.

\bibitem[HSV05]{hemaspaandra2005complexity}
Edith Hemaspaandra, Holger Spakowski, and J{\"o}rg Vogel.
\newblock The complexity of kemeny elections.
\newblock {\em Theoretical Computer Science}, 349(3):382--391, 2005.

\bibitem[Kem59]{kemeny1959mathematics}
John~G Kemeny.
\newblock Mathematics without numbers.
\newblock {\em Daedalus}, 88(4):577--591, 1959.

\bibitem[Lac14]{Lackner14}
Martin Lackner.
\newblock Incomplete preferences in single-peaked electorates.
\newblock In {\em Proceedings of the Twenty-Eighth {AAAI} Conference on
  Artificial Intelligence (AAAI)}, pages 742--748, 2014.

\bibitem[LB11a]{LuB11a}
Tyler Lu and Craig Boutilier.
\newblock Robust approximation and incremental elicitation in voting protocols.
\newblock In {\em {IJCAI} 2011, Proceedings of the 22nd International Joint
  Conference on Artificial Intelligence (IJCAI)}, pages 287--293, 2011.

\bibitem[LB11b]{lu2011vote}
Tyler Lu and Craig Boutilier.
\newblock Vote elicitation with probabilistic preference models: Empirical
  estimation and cost tradeoffs.
\newblock In {\em Algorithmic Decision Theory}, pages 135--149. 2011.

\bibitem[Lev75]{levenglick1975fair}
Arthur Levenglick.
\newblock Fair and reasonable election systems.
\newblock {\em Behavioral Science}, 20(1):34--46, 1975.

\bibitem[MF14]{magiera2014hard}
Krzysztof Magiera and Piotr Faliszewski.
\newblock How hard is control in single-crossing elections.
\newblock {\em Proceedings of the 21st European Conference on Artificial
  Intelligence (ECAI)}, 2014.

\bibitem[Mir71]{mirrlees1971exploration}
James~A Mirrlees.
\newblock An exploration in the theory of optimum income taxation.
\newblock {\em The review of economic studies}, pages 175--208, 1971.

\bibitem[Mou91]{moulin1991axioms}
Hervi Moulin.
\newblock {\em Axioms of cooperative decision making}.
\newblock Number~15. Cambridge University Press, 1991.

\bibitem[PHG00]{pennock2000social}
David~M. Pennock, Eric Horvitz, and C.~Lee Giles.
\newblock Social choice theory and recommender systems: Analysis of the
  axiomatic foundations of collaborative filtering.
\newblock In {\em Proceedings of the 17th International Conference on
  Artificial Intelligence (AAAI)}, 2000.

\bibitem[PRZ08]{procaccia2008complexity}
Ariel~D Procaccia, Jeffrey~S Rosenschein, and Aviv Zohar.
\newblock On the complexity of achieving proportional representation.
\newblock {\em Social Choice and Welfare}, 30(3):353--362, 2008.

\bibitem[Rob77]{roberts1977voting}
Kevin~WS Roberts.
\newblock Voting over income tax schedules.
\newblock {\em Journal of Public Economics}, 8(3):329--340, 1977.

\bibitem[Sat75]{satterthwaite1975strategy}
Mark~Allen Satterthwaite.
\newblock Strategy-proofness and arrow's conditions: Existence and
  correspondence theorems for voting procedures and social welfare functions.
\newblock {\em Journal of Economic Theory}, 10(2):187--217, 1975.

\bibitem[ST06]{saporiti2006single}
Alejandro Saporiti and Fernando Tohm{\'e}.
\newblock Single-crossing, strategic voting and the median choice rule.
\newblock {\em Social Choice and Welfare}, 26(2):363--383, 2006.

\bibitem[SYFE13]{skowron2013complexity}
Piotr Skowron, Lan Yu, Piotr Faliszewski, and Edith Elkind.
\newblock The complexity of fully proportional representation for
  single-crossing electorates.
\newblock In {\em Symposium Algorithmic Game Theory (SAGT)}, pages 1--12.
  Springer, 2013.

\bibitem[Wal07]{walsh2007uncertainty}
Toby Walsh.
\newblock Uncertainty in preference elicitation and aggregation.
\newblock In {\em AAAI}, volume~7, pages 3--8, 2007.

\bibitem[You77]{young1977extending}
H~Peyton Young.
\newblock Extending condorcet's rule.
\newblock {\em Journal of Economic Theory}, 16(2):335--353, 1977.

\end{thebibliography}
